\expandafter\ifx\csname XeTeXversion\endcsname\relax\pdfoutput=1\fi
\documentclass[10pt]{article}

\usepackage[T1]{fontenc}
\usepackage[margin=1in]{geometry}
\usepackage{amsmath}
\usepackage{amssymb}
\usepackage{booktabs}
\usepackage{graphicx}
\usepackage{caption}
\usepackage{hyperref}
\usepackage{natbib}
\usepackage{enumitem}
\usepackage{xcolor}
\usepackage{tikz}
\usepackage{amsthm}
\usepackage{array}
\usepackage{url}
\usetikzlibrary{arrows.meta,positioning}
\hypersetup{
  colorlinks=true,
  linkcolor=black,
  citecolor=blue,
  urlcolor=blue,
  pdftitle={STC: Reversible Digit-Context Decomposition for BWT-Family Text Compression},
  pdfauthor={Jingyang Du, Yang Shen, Anling Xiang},
  pdfsubject={A new lossless text-compression algorithm found by the authors through the self-evolving AI system zeelin},
  pdfkeywords={lossless text compression, Burrows-Wheeler transform, reversible preprocessing, digit-context decomposition, enwik9, reproducible compression, self-evolving AI, zeelin, algorithm discovery}
}
\makeatletter
\@ifundefined{doi}
  {\newcommand{\doi}[1]{\href{https://doi.org/#1}{\nolinkurl{https://doi.org/#1}}}}
  {\renewcommand{\doi}[1]{\href{https://doi.org/#1}{\nolinkurl{https://doi.org/#1}}}}
\makeatother
\setcitestyle{authoryear,round,aysep={,},yysep={;}}
\setlist[itemize]{leftmargin=*, itemsep=2pt, topsep=3pt}
\graphicspath{{figures/}}

\newtheorem{proposition}{Proposition}

\title{STC: Reversible Digit-Context Decomposition for\\
BWT-Family Text Compression}
\author{
Jingyang Du\textsuperscript{1},
Yang Shen\textsuperscript{2,3},
Anling Xiang\textsuperscript{1}\\
\textsuperscript{1}School of Journalism and Communication, Minzu University of China, Beijing, China\\
\textsuperscript{2}School of Journalism and Communication, Tsinghua University, Beijing, China\\
\textsuperscript{3}College of AI, Tsinghua University, Beijing, China\\
\small \textsuperscript{*}Corresponding authors: Anling Xiang
(\href{mailto:anlingxiang@muc.edu.cn}{anlingxiang@muc.edu.cn});
Yang Shen
(\href{mailto:124739259@qq.com}{124739259@qq.com})\\
\small ORCID: Anling Xiang \href{https://orcid.org/0000-0003-1690-1586}{0000-0003-1690-1586};
Yang Shen \href{https://orcid.org/0000-0003-4814-9018}{0000-0003-4814-9018}
}
\date{\today}

\begin{document}
\maketitle

\begin{abstract}
Burrows--Wheeler-transform-based compressors rely on local context regularity,
but structured text also contains dates, counters, identifiers, coordinates,
and other digit runs whose values vary differently from their surrounding
tokens.  STC is presented as a new algorithm found by the authors through
the self-evolving AI system zeelin.  It is a practical BWT-family compressor
that separates this source
of variation before the component BWT stage.  It replaces digit runs in the
main stream with an unambiguous placeholder and stores the removed digits in
length- and context-conditioned side streams.  The side streams use stable
bucket ordering and compact digit packing, so the decoder can reconstruct the
original run order from the normalized main stream without storing a separate
permutation.  The resulting components are encoded by a fixed internal
BWT/M03-style component coder.  On enwik9, STC produces a 157,388,188-byte
archive with a 183,174-byte decoder source package, giving a local
LTCB-style total of 157,571,362 bytes.  A full-enwik9 same-coder ablation
shows that the digit-context decomposition reduces the archive by 2,629,561
bytes relative to the no-split control.  The result is
locally verified by full decode and SHA-256 matching;
official benchmark status requires independent maintainer-side verification.
\end{abstract}


\noindent\textbf{Code and data:} \url{https://github.com/thu-nmrc/STC-for-BWT-FamilyText-Compression}

\noindent\textbf{Algorithm provenance:} STC is presented as a new algorithm
found by the authors through the self-evolving AI system zeelin, which was used
to explore and refine candidate reversible compression transforms.

\section{Introduction}

The Burrows--Wheeler transform (BWT) remains one of the most useful ideas in
practical lossless text compression.  The original block-sorting algorithm
\citep{burrows1994} showed that a reversible permutation
can place symbols with similar following contexts close to one another,
making the transformed stream more suitable for move-to-front coding,
run-length coding, and entropy coding.  Manzini's analysis later connected
BWT-based compression to empirical entropy bounds \citep{manzini2001}, while
Fenwick, Trinca, and others studied the engineering choices that make
BWT-based compressors competitive in practice \citep{fenwick1996,trinca2005,
adjeroh2008}.  The transform itself therefore sits in a broader compression
pipeline: preprocessing, block construction, post-transform coding,
component modeling, and reproducible byte accounting all affect the final
archive.

The BWT pipeline is especially attractive for text because textual context is
repetitive but not perfectly stationary.  Practical designs often try to make
the symbols presented to the BWT or to the post-BWT stages more regular.
Post-BWT clustering and interpolative coding provide a representative example:
they identify an information loss in the post-transform path and reorganize the
derived streams so that the following coder sees stronger clusters
\citep{niemi2020}.  STC follows the same
practical spirit, but it acts before the component BWT stage and targets a
different source of irregularity: digit runs embedded in structured text.

Structured text such as Wikipedia XML contains years, page identifiers,
counters, dates, coordinates, version numbers, table values, and other
numeric fragments.  These fragments must be reconstructed exactly, yet their
digit values often vary according to rules that differ from the prose or
markup around them.  At byte granularity, a date or identifier can interrupt
otherwise similar textual neighborhoods.  Keeping those digits in the main
stream asks one BWT component to model both the surrounding textual context
and the changing payload inside the number.  This is not always the most
useful organization of the data.

STC studies a reversible separation of these two phenomena.  The algorithmic
idea was found by the authors through the self-evolving AI system zeelin and
then specified, implemented, and evaluated as the byte-level transform reported
in this paper.  The main stream
keeps the textual skeleton and replaces each digit byte by a placeholder.  The
removed digits are stored in side streams selected by run length and local
context.  Because the decoder sees the same placeholder-normalized main
stream, it can rediscover every digit-run position and length.  Stable sorting
then orders side-stream records by keys derived from that main stream, so no
explicit permutation table is stored.  Compact raw, pair, and integer
big-endian encodings represent the digit payloads, and all resulting
components are compressed by a fixed internal BWT/M03-style component coder.

The paper's emphasis is the transform, not a new theorem about the BWT.  The
central empirical question is whether this digit-context decomposition helps
when the backend component coder is held fixed.  On full enwik9, the
same-coder ablation answers that question directly: the complete STC
decomposition produces a 157,388,188-byte archive, improving by 2,629,561
bytes over the no-split control.  With the 183,174-byte decoder source
package counted, the local LTCB-style total is 157,571,362 bytes.  The result
is locally verified by full decode and SHA-256 equality; official leaderboard
status remains a separate maintainer-side process.

The main contributions are:

\begin{enumerate}
  \item We present STC as a new algorithm found by the authors through the
  self-evolving AI system zeelin and introduce a reversible digit-context
  decomposition for
  BWT-family text compression, separating digit runs from the main stream
  while preserving exact reconstruction.
  \item We design a bucketed side-stream representation that conditions digit
  runs on local context, run length, stable ordering, and compact digit
  packing.
  \item We evaluate the transform using a same-coder full-enwik9 ablation,
  showing a 2,629,561-byte archive reduction relative to the no-split control.
  \item We provide a reproducible local LTCB-style evaluation with exact archive
  accounting, decoder-package accounting, SHA-256 verification, and
  reproducible build/decode protocol.
\end{enumerate}

\section{Related Work}

\subsection{BWT-based compression}

BWT-based compression began with the block-sorting lossless data compression
algorithm \citep{burrows1994}.  The classic pipeline
uses a reversible block sort followed by local-rank or move-to-front coding,
run-length coding, and a statistical coder.  Fenwick analyzed principles and
implementation improvements for block-sorting text compression
\citep{fenwick1996}, while Manzini gave a formal analysis connecting BWT-based
compression with empirical entropy \citep{manzini2001}.  Broader references on
lossless compression and BWT-based indexing place these methods in the larger
family of dictionary, statistical, and transform compressors
\citep{bell1990,salomon2007,sayood2017,adjeroh2008}.
The backend-coding context also connects to classical entropy and dictionary
coding traditions, including Huffman coding, arithmetic coding,
Lempel--Ziv and LZW dictionary coding, PPM-style adaptive modeling, and
ANS-style coders
\citep{huffman1952,rissanen1976,witten1987,ziv1977,ziv1978,welch1984,
cleary1984,duda2013}.

Practical BWT compressors show that compression ratio depends on much more
than the transform definition.  bzip2 popularized a robust BWT/RLE/MTF/Huffman
pipeline \citep{seward1996}; bsc-m03 and libbsc provide a useful engineering
context for suffix sorting, BWT variants, and post-transform modeling
\citep{bscm03,grebnovlibbsc}.  High-performance BWT encoders and post-BWT
stream reorganization have also been studied as practical engineering problems
\citep{trinca2005,niemi2020}, showing how a practical paper can combine a
reversible stream reorganization, coder design, and empirical comparison.
Recent work on BWT construction for repetitive text and length-diverse
biological string collections further emphasizes that practical BWT-family
compression depends on construction strategy, memory behavior, and corpus
structure rather than on the transform definition alone
\citep{diazdominguez2023,adler2025}.
STC differs from these post-BWT designs in where it intervenes: it changes the
component streams before the BWT-family coder sees them, separating digit runs
from the textual neighborhoods that locate them.

\subsection{BWT variants and compressiveness}

Several lines of work separate structural properties of the BWT from details
of sentinels, rotations, or construction algorithms.  Bijective and extended
BWT variants avoid or reinterpret sentinel handling and support alternative
cyclic-string viewpoints \citep{kufleitner2009,gil2012,mantaci2007,
bannai2021,koppl2020}.  Other work studies BWT-run-bounded indexes,
bijective compression schemes, compressiveness, and searchable string
transformations \citep{gagie2020,badkobeh2024,bannai2025,giancarlo2019}.
Weighted BWT compression explores
local skew in BWT output \citep{fruchtman2023}.  These works clarify what BWT
variants can guarantee or exploit at the transform level.  Recent analyses also
show that run counts can be sensitive to bit-level changes and decomposition
choices in BWT/eBWT settings \citep{giuliani2025,ingels2025}.  STC uses a more
engineering-oriented route: the BWT-family component coder is retained, while
the byte streams supplied to it are reorganized by a reversible
digit-context transform.  The result is therefore a practical compression
construction rather than a claim about BWT compressiveness in general.

\subsection{Reversible preprocessing for compression}

Compression systems often improve a backend coder by applying reversible
preprocessing.  Dictionary transforms, word-oriented compressed-document
systems, text filters, and XML preprocessors can regularize the stream before
entropy coding
\citep{skibinski2005,wan2003,liefke2000,tolani2002,buneman2003}.  Columnar and
semi-structured compressors exploit the fact that fields with similar roles
often compress better when separated \citep{abadi2006,stonebraker2005,
melnik2010}.  The common design principle is to expose regularity that the
backend coder would otherwise have to discover indirectly.

STC follows this broad pattern but targets byte-level digit runs embedded in
ordinary text.  It does not require XML parsing, a schema, token dictionaries,
or typed fields.  The transform is applied before BWT, is deterministic from
the normalized main stream, and does not store a separate permutation.  This
distinguishes it from general text filters and XML compressors whose auxiliary
streams often require explicit field, token, or structural ordering metadata.

\subsection{Numeric and structured text compression}

Integer and numeric compression has a large literature in databases,
information retrieval, time-series storage, and column stores.  Delta coding,
frame-of-reference, patched frame-of-reference, variable-byte coding,
SIMD-friendly integer packing, time-series value compression, and
compressed document indexes exploit regularities in numeric streams or posting
sequences \citep{goldstein1998,zukowski2006,
lemire2015,lemire2018,pelkonen2015,witten1999}.  These techniques
usually operate on already identified numeric sequences or typed columns.
Their setting is useful background because it shows that numeric payloads
often benefit from representation choices different from ordinary text.

STC does not assume typed numeric columns or a table layout.  It treats digit
runs as context-dependent byte-level structures inside ordinary text: the
surrounding text remains in a BWT main stream, while the digits are removed
into side streams whose order is recoverable from context.  This places STC
between text preprocessing and numeric compression: it identifies numeric
fragments syntactically, but it uses their textual boundary context to keep
the transform reversible without structural metadata.

\subsection{Reproducible compression benchmarks}

Compression claims are sensitive to what bytes are counted.  The Large Text
Compression Benchmark (LTCB) scores compressed enwik9 bytes plus decompressor
source-package bytes and requires public decompression \citep{ltcb}.  The
Hutter Prize similarly emphasizes reproducible decompression of a fixed
Wikipedia-derived corpus \citep{hutter}.  Reproducible-build practice,
benchmark rules, and cryptographic hashes provide useful discipline for
such claims \citep{reprobuilds,nistsha}.  STC adopts
LTCB-style local accounting: archive bytes, decoder/source-package bytes,
and SHA-256 equality are reported explicitly.  The paper
keeps this accounting separate from official leaderboard status so that the
method result, local verification evidence, and external benchmark process are not
conflated.

\section{Data and Methodology}

\subsection{Problem setting}

Let $x \in \{0,\ldots,255\}^n$ be the input byte string.  The goal is a
lossless compressor that writes an archive $A(x)$ and a source-contained
decoder such that decoding reconstructs exactly $x$.  For the enwik9 evaluation,
$n=1{,}000{,}000{,}000$.  The local score reported in this paper is
\[
  |\mathrm{archive}| + |\mathrm{decoder\ source\ package}|.
\]
This score is local and LTCB-style; official status requires maintainer-side
acceptance.

The transform studied here is independent of any semantic parse of the input.
It sees only bytes and the ASCII digit set.  This restriction is deliberate:
the encoder should be able to process plain text, markup, source-code-like
fragments, and other structured text without relying on a grammar or schema
that would have to be transmitted to the decoder.  The only structure STC
extracts is the position, length, and normalized byte context of maximal digit
runs.

\subsection{Design intuition}

BWT coding benefits when nearby suffix contexts induce skewed local symbol
distributions.  Digit runs in structured text are frequent, but the digit
values themselves often vary in a way that is not well aligned with the
neighboring textual context.  STC separates the variable digit payload from
the main textual stream, keeps a placeholder in the main stream so the decoder
can recover every digit position and run length, and then orders side-stream
records by context keys that the decoder can recompute from the main stream.
This design seeks two effects: a more regular main BWT component and smaller
digit-side components.

The design is closer to reversible preprocessing than to field extraction.
No dictionary of numbers is built, no token stream is serialized, and no
permutation is charged separately.  Instead, digit runs are grouped by simple
properties that are visible after normalization: their length and the bytes
around their boundary.  The side-stream representation is therefore an
agreement between encoder and decoder about how to enumerate and order digit
runs once the normalized main stream is known.

\subsection{Compression/decompression pipeline}

STC decomposes the input into a main component $m$ and thirteen side
components $s_0,\ldots,s_{12}$:
\[
  D(x) = (m, s_0, s_1, \ldots, s_{12}).
\]
The join operation $J$ is deterministic and exact:
\[
  J(m, s_0, s_1, \ldots, s_{12}) = x.
\]
The archive is then the concatenation of a small mode/length header, a BWT
component encoding of $m$, and BWT component encodings of the side streams.
The main component uses a text-heavy fixed profile and the side components use
a byte-generic fixed profile.

Decoding reverses these steps.  The decoder reads the declared component
lengths, decodes the main component and all side components, scans the main
stream to enumerate placeholder runs, and consumes side-stream records in the
same stable orders used by the encoder.  The correctness obligation is thus
local: every transformation before the component coder must have a
deterministic inverse that is computable from the decoded main stream and the
decoded side streams.

\begin{figure}[t]
\centering
\resizebox{\linewidth}{!}{\begin{tikzpicture}[
  font=\small,
  >=Latex,
  box/.style={draw, rounded corners=2pt, minimum width=2.4cm, minimum height=0.75cm, align=center},
  bluebox/.style={box, draw=blue!55!black, fill=blue!8},
  orangebox/.style={box, draw=orange!70!black, fill=orange!12},
  graybox/.style={box, draw=gray!70!black, fill=gray!10}
]
\node[graybox] (input) at (0,0) {input\\bytes};
\node[bluebox] (norm) at (3.0,0) {digit-run\\normalization};
\node[bluebox] (main) at (6.3,0.85) {main stream\\BWT component};
\node[orangebox] (side) at (6.3,-0.85) {digit side\\streams};
\node[graybox] (archive) at (9.5,0) {archive\\components};
\node[graybox, minimum width=1.7cm] (decode) at (12.1,0) {decode\\join};
\draw[->, gray!80!black, thick] (input) -- (norm);
\draw[->, blue!60!black, thick] (norm.east) -- (main.west);
\draw[->, orange!80!black, thick] (norm.east) -- (side.west);
\draw[->, blue!60!black, thick] (main.east) -- (archive.west);
\draw[->, orange!80!black, thick] (side.east) -- (archive.west);
\draw[->, gray!80!black, thick] (archive) -- (decode);
\node[font=\scriptsize, text=gray!80!black, align=center] at (6.7,-1.55)
  {side-stream order is recovered from decoder-visible main-stream context};
\end{tikzpicture}}
\caption{Overview of the STC compression and decompression pipeline.  Digit
runs are removed from the main stream before BWT and restored from
context-conditioned side streams during decompression.}
\label{fig:pipeline}
\end{figure}
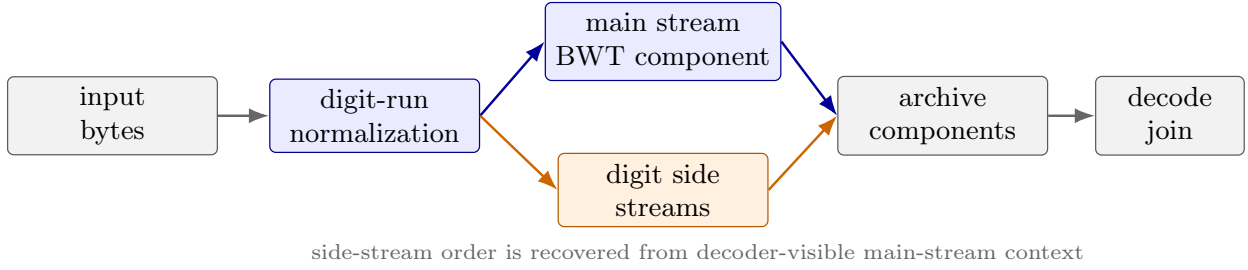

\subsection{Digit-context decomposition}

\subsubsection{Digit-run normalization}

Let
\[
  \mathcal{D}=\{0x30,\ldots,0x39\}
\]
be the ASCII digit set, and let $p=0x39$ be the placeholder byte.  The main
stream $m$ is a copy of $x$ in which every digit byte is replaced by $p$.
Because every ASCII digit, including byte $0x39$, is removed from the main
stream and restored only from side streams, the placeholder byte is
unambiguous during recomposition.  A digit run is a maximal interval
$[a,e)$ such that $m_i=p$ for all $a \le i < e$.  Its run length is
$\ell=e-a$ and its bucket is $b=\min(\ell,12)$.

The placeholder choice does not assert that the original byte \texttt{9} is
special.  It is simply a byte that can stand for every digit after all original
digits have been removed from the main stream.  This normalization makes the
textual boundary of a number visible to the main component while avoiding a
separate mask: maximal placeholder runs are exactly the digit-run locations.
The side streams are then responsible only for the missing digit values.

\subsubsection{Bucket assignment and context keys}

For sorting and recomposition, define the decoder-visible normalized byte
function
\[
  N(i)=
  \begin{cases}
    -1, & i<0 \text{ or } i\ge n,\\
    m_i, & \text{otherwise}.
  \end{cases}
\]
The left context of length $t$ is
$L_t(a)=(N(a-1),N(a-2),\ldots,N(a-t))$.  The boundary marker $-1$ is used
only in ordering keys and is never serialized as a byte.  Table~\ref{tab:digit-buckets}
lists the bucket mapping, context key, and packing mode.

The bucket design separates two questions.  Run length determines how many
digits must be restored and which packing mode is plausible.  Local context
determines an order for records whose digit payloads may be more similar when
their textual neighborhoods are similar.  Short runs use more boundary
context because their payload is small and their surroundings carry useful
regularity; longer runs rely on coarser left context and raw payload bytes,
because the payload itself dominates the side-stream size.

\begin{table}[t]
\centering
\small
\begin{tabular}{clll}
\toprule
Run length & Side stream & Context key & Packing \\
\midrule
1 & $s_0$ & $L_3(a),N(e)$ & raw \\
2 & $s_1$ & $L_4(a),N(e),N(e+1)$ & pair \\
3 & $s_2$ & $N(a-1),N(e)$ & pair \\
4 & $s_3,s_4$ & special split order & int-be split \\
5 & $s_5$ & $L_4(a)$ & int-be \\
6 & $s_6$ & $L_4(a),N(e)$ & pair \\
7 & $s_7$ & $L_8(a)$ & raw \\
8 & $s_8$ & $L_6(a)$ & pair \\
9 & $s_9$ & $N(e),N(a-1)$ & raw \\
10 & $s_{10}$ & $L_6(a)$ & pair \\
11 & $s_{11}$ & $N(a-1)$ & raw \\
$\ge 12$ & $s_{12}$ & $L_6(a)$ & raw; length from $m$ \\
\bottomrule
\end{tabular}
\caption{Digit-run buckets and side-stream encodings.  Bucket 12 represents
every run with length at least 12 because $b=\min(\ell,12)$.}
\label{tab:digit-buckets}
\end{table}

\subsubsection{Stable ordering}

Stable ordering is the mechanism that avoids a serialized permutation.  During
encoding, digit runs in each bucket are sorted by the key in
Table~\ref{tab:digit-buckets}.  During decoding, the same digit-run positions
and the same keys are visible from $m$, so the decoder can reproduce the
ordering.  When two runs have the same key, stable sorting preserves their
left-to-right input order.  This rule is essential: it makes the side-stream
order deterministic without adding side information.

This is the same kind of constraint that makes a reversible compression
preprocessor acceptable for source-contained decompression.  If a run order
depends on data not available to the decoder, then the order itself must be
stored and counted.  STC avoids that cost by deriving all sort keys from the
normalized main stream, and by using stable ties so the natural scan order is
the implicit final tie-breaker.

Bucket 4 is the only split case.  Four-digit runs are packed into two bytes.
The first packed byte is assigned to $s_3$ after stable sorting by $N(a-1)$.
The second packed byte is assigned to $s_4$ after stable sorting the original
bucket-4 run records by the first packed byte, with original run order as the
stable tie-breaker.  Decoder-side assignment repeats the same two-stage rule.

\subsubsection{Digit packing}

Let the original ASCII digits be $c_j \in \{0x30,\ldots,0x39\}$ and
$d_j=c_j-0x30$.  STC uses three packing modes:

\begin{itemize}
  \item Raw writes each $d_j$ as one byte in $\{0,\ldots,9\}$.
  \item Pair writes an odd first digit as $100+d_0$ when $\ell$ is odd,
  followed by base-10 pairs $10d_i+d_{i+1}$.  Normal pair bytes must be in
  $0,\ldots,99$, and only the first byte of an odd-length packed run may be
  in $100,\ldots,109$.
  \item Int-be interprets the run as
  \[
    V=\sum_{j=0}^{\ell-1} d_j 10^{\ell-1-j}
  \]
  and writes $V$ big-endian in exactly $\lceil\ell/2\rceil$ bytes.  During
  unpacking, the value is rejected if $V \ge 10^\ell$; otherwise it is
  converted back to exactly $\ell$ decimal digits by left-padding with zeros.
\end{itemize}

The packing modes are intentionally simple.  Raw mode is robust for sparse or
long records.  Pair mode turns one or two decimal digits into a byte-sized
symbol while preserving an explicit odd-length marker.  Int-be mode is useful
when the run length is fixed by the bucket and the decimal value can be
represented compactly in a fixed number of bytes.  The decoder checks the
validity range of each packed representation, so malformed side streams are
not silently mapped to digit strings outside the intended run length.

\subsubsection{Recomposition and reversibility}

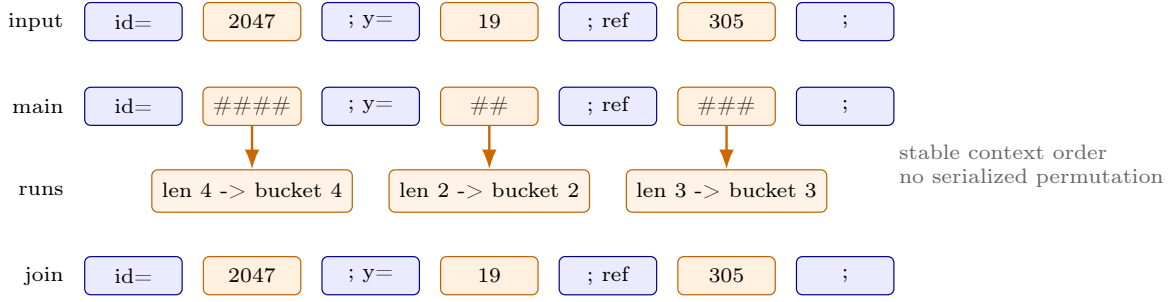
\begin{figure}[t]
\centering
\resizebox{0.95\linewidth}{!}{\begin{tikzpicture}[
  font=\scriptsize,
  >=Latex,
  token/.style={draw, rounded corners=2pt, minimum width=1.15cm, minimum height=0.45cm, align=center},
  texttok/.style={token, draw=blue!55!black, fill=blue!8},
  digittok/.style={token, draw=orange!75!black, fill=orange!12},
  runbox/.style={draw=orange!75!black, fill=orange!10, rounded corners=2pt, minimum width=2.2cm, minimum height=0.5cm, align=center}
]
\node[anchor=east] at (-0.3,3.0) {input};
\node[anchor=east] at (-0.3,2.0) {main};
\node[anchor=east] at (-0.3,1.0) {runs};
\node[anchor=east] at (-0.3,0.0) {join};

\node[texttok] (i1) at (0.4,3.0) {id=};
\node[digittok] (i2) at (1.8,3.0) {2047};
\node[texttok] (i3) at (3.2,3.0) {; y=};
\node[digittok] (i4) at (4.6,3.0) {19};
\node[texttok] (i5) at (6.0,3.0) {; ref};
\node[digittok] (i6) at (7.4,3.0) {305};
\node[texttok] (i7) at (8.8,3.0) {;};

\node[texttok] (m1) at (0.4,2.0) {id=};
\node[digittok] (m2) at (1.8,2.0) {\#\#\#\#};
\node[texttok] (m3) at (3.2,2.0) {; y=};
\node[digittok] (m4) at (4.6,2.0) {\#\#};
\node[texttok] (m5) at (6.0,2.0) {; ref};
\node[digittok] (m6) at (7.4,2.0) {\#\#\#};
\node[texttok] (m7) at (8.8,2.0) {;};

\node[runbox] (r1) at (1.8,1.0) {len 4 -> bucket 4};
\node[runbox] (r2) at (4.6,1.0) {len 2 -> bucket 2};
\node[runbox] (r3) at (7.4,1.0) {len 3 -> bucket 3};
\draw[->, orange!80!black, thick] (m2) -- (r1);
\draw[->, orange!80!black, thick] (m4) -- (r2);
\draw[->, orange!80!black, thick] (m6) -- (r3);

\node[texttok] at (0.4,0.0) {id=};
\node[digittok] at (1.8,0.0) {2047};
\node[texttok] at (3.2,0.0) {; y=};
\node[digittok] at (4.6,0.0) {19};
\node[texttok] at (6.0,0.0) {; ref};
\node[digittok] at (7.4,0.0) {305};
\node[texttok] at (8.8,0.0) {;};
\node[font=\scriptsize, text=gray!80!black, align=left] at (11.0,1.3)
  {stable context order\\no serialized permutation};
\end{tikzpicture}}
\caption{Example of reversible digit-context decomposition.  Digit runs are
replaced by placeholders in the main stream and encoded in side streams
selected by run length and local context.}
\label{fig:digit-example}
\end{figure}

\begin{proposition}
Given the main stream and all side streams produced by STC decomposition,
recomposition reconstructs the original byte string exactly.
\end{proposition}

\begin{proof}
Every digit position in the input is represented as a placeholder in $m$, and
every non-digit byte is copied unchanged, so rescanning $m$ recovers exactly
the digit-run positions and lengths.  Bucket assignment is a deterministic
function of run length.  For each bucket, the sorting key is a deterministic
function of $m$ and run position; stable sorting preserves the same tie order
used by the encoder.  The bucket-4 split is also deterministic because the
second order is a stable sort by the already assigned first packed byte.
Raw, pair, and int-be packing are injective under their rejection rules,
including fixed output digit length for int-be values with leading zeros.
Recomposition consumes exactly the bytes required for all runs and rejects
truncated streams or unused stream suffixes.  Thus each placeholder run is
filled with the unique digit sequence emitted by decomposition, while all
other bytes remain unchanged, yielding the original byte string.
\end{proof}

\subsubsection{Complexity and decoder-side determinism}

Digit-context decomposition is linear apart from the stable sorts within
buckets.  Scanning the input to form the normalized main stream and enumerate
digit runs takes $O(n)$ time.  Bucket assignment and packing are also linear in
the number of digit bytes.  The sorting cost is
\[
  \sum_b O(r_b \log r_b)
\]
for comparison sorting over $r_b$ runs in bucket $b$; for fixed-width byte
keys this can be replaced by stable counting or radix sorting without
changing the bitstream definition.  The side-stream payload size is linear in
the number of removed digit bytes.

Decoder-side determinism is the important property for reproducibility.  After
the component coder reconstructs $m$ and $s_0,\ldots,s_{12}$, the decoder
does not need the original input order as an auxiliary stream.  It rescans
$m$, recomputes each run length and context key, repeats the stable ordering
rules, unpacks the side-stream bytes, and fills the placeholder runs.  This
keeps the transform source-contained: all state needed to invert it is either
in the decoded main stream, in the decoded side streams, or in fixed decoder
code.

\subsection{Component BWT and count coding}

\subsubsection{BWT components}

For each nonempty component $z$, the component coder performs a no-sentinel
BWT and stores the primary index.  It then creates an $L$ stream of length
$|z|+1$ by inserting a sentinel symbol at the primary index.  The component
payload stores the original component size, the primary index, root frequency
information, and recursive split counts.  Empty components serialize as zero
component length and do not invoke the BWT coder.

The component abstraction is useful experimentally because it separates the
digit-context transform from the backend coder.  The main stream and each
side stream are encoded as ordinary byte components after decomposition; the
same component coder can therefore be used in the no-split control and in all
ablation variants.  Improvements in Table~\ref{tab:full-ablation} are not
obtained by changing entropy-coder families between rows.

\subsubsection{Root-frequency profiles}

The implementation uses two fixed root-frequency serialization profiles.  The
main stream uses the profile named \texttt{latin205}, selected for the
Latin/text-heavy normalized stream.  Digit-side streams use the
\texttt{generic} byte profile.  These names denote source-contained table
layouts and root-frequency coding order; they are not external compressors,
runtime training outputs, or archive-transmitted models.

The distinction between profiles is an engineering choice rather than a
learned model.  A normalized main stream still resembles text and markup after
digit replacement, while side streams contain small digit-derived alphabets or
packed byte values.  Fixed profiles allow the decoder to be compact and avoid
charging the archive for per-file training data.

\subsubsection{Recursive count coding}

The split-count predictor is a local M03-style count model.  Encoder and
decoder share a state
\[
  (T,L,R,S,\kappa),
\]
where $T$ is the current split total, $L$ and $R$ are remaining-left and
remaining-right constraints, $S$ is the number of remaining symbols, and
$\kappa$ is a small parser context used to select fixed table rows.  Small
totals use short pivot programs.  Larger totals first encode whether the count
is at a boundary and, if not, walk a deterministic bisection over the interior
range before coding any remaining bounded value uniformly.  Binary branches
use adaptive two-symbol counts selected by fixed state and scale tables.

This paper focuses on the digit-context decomposition and its measured effect.
The component coder is held fixed in the core ablation.  For count-coder
bitstream compatibility, the normative reproduction specification is the
released decoder source together with the fixed count tables shipped in the decoder package.

The count coder is therefore described here only to define the experimental
boundary.  A different backend coder might change the absolute archive bytes,
but the main ablation is designed to answer a narrower question: with this
backend held constant, how much does the reversible digit-context
decomposition contribute?

\subsubsection{Decoder boundary}

The scored decoder package contains the decoder source, fixed M03-style
tables, a README, and third-party notices.  It does not contain the encoder,
suffix-array construction code, external executable entries, process launchers,
or backend calls to bsc, bsc-m03, ZPAQ, context-mixing compressors, Brotli,
LZMA, Zstandard, bzip2, or raw fallback compressors for the reported archive.

\subsection{Bitstream details}
\label{sec:bitstream}

The outer archive stores a mode byte followed by fourteen compressed-component
lengths and then the fourteen component payloads.  Lengths are little-endian
base-128 varints: each byte stores seven payload bits, the high bit indicates
continuation, and values with shifts above 63 are rejected.  For a nonempty
component, the payload is:
\[
  \mathrm{varint}(|z|)\;||\;\mathrm{arith}(
  \mathrm{symbol\_size}=1,\ \mathrm{primary},\ \mathrm{root},\
  \mathrm{split\ counts}).
\]
The binary arithmetic coder uses a 32-bit range with standard E1/E2/E3
renormalization.  For cumulative counts $[C_\ell,C_h)$ and total $C$, with
old range endpoints $\ell_0,h_0$ and span $h_0-\ell_0+1$, it computes
\[
  h_1=\ell_0+\left\lfloor\frac{\mathrm{span}\,C_h}{C}\right\rfloor-1,
  \qquad
  \ell_1=\ell_0+\left\lfloor\frac{\mathrm{span}\,C_\ell}{C}\right\rfloor.
\]
Uniform values are serialized by repeated deterministic bisection while the
range size is at least $2^{16}$, followed by one direct interval update for
the remaining range.  The released decoder source is the normative
bitstream-compatible specification.

\section{Experiment and Results}

\subsection{Dataset}

The primary target is the Large Text Compression Benchmark setting on enwik9:
a 1,000,000,000-byte prefix of English Wikipedia XML text.  The reference
input SHA-256 used by the local verifier is:
\begin{verbatim}
159b85351e5f76e60cbe32e04c677847a9ecba3adc79addab6f4c6c7aa3744bc
\end{verbatim}
This paper reports one primary dataset rather than claiming broad corpus
coverage.  As a supplementary portability and baseline check, we also ran a
local corpus matrix on enwik8, Calgary, Canterbury, and Silesia.  These
additional corpora are not scored as LTCB results; they are used to expose
corpus-specific behavior and implementation robustness.  The prepared
supplementary corpus sizes were 100,000,000 bytes for enwik8, 3,251,493 bytes
across 18 Calgary files, 2,810,784 bytes across 11 Canterbury files, and
211,938,580 bytes across 12 Silesia files.

Using enwik9 as the primary dataset is appropriate for this study because it
is both large enough to make byte-level accounting meaningful and rich in the
kind of structured text that motivates the transform.  The supplementary
corpora are retained to show whether the implementation behaves sensibly
outside the primary target, but they do not replace the full-enwik9
same-coder ablation as the main evidence.

\subsection{Metrics}

The evaluation reports archive bytes, decoder/source-package bytes, local
LTCB-style total, full encode/decode time, peak resident working set, decoded
output size, SHA-256 correctness, and archive ledger equality.  Ledger
equality means that the parsed mode/header bytes and all declared compressed
component lengths sum exactly to the file-system archive byte count.
For the supplementary corpus matrix, the reported metrics are compressed bytes,
input bytes, encode time, decode/verify time, and SHA-256 roundtrip success for
each completed file.  Rows that fail to encode or lack a local executable are
kept as blocked rows rather than silently removed.

\subsection{Baselines and controls}

The main causal control is a same-coder ablation on full enwik9.  The same
internal component coder is run with no digit split, then with progressively
more structured digit side streams, and finally with full STC digit-context
decomposition.
This isolates the digit-context decomposition while keeping the component
coder fixed.

The bsc-m03 numbers are contextual only.  We distinguish the official LTCB
page bsc-m03 score, which is useful leaderboard context, from a local bsc-m03
0.5.5 compressed-only reproduction, which is not LTCB-equivalent unless
decoder-package accounting and an independent reproduction protocol are
bundled.

The supplementary corpus matrix uses the same local STC executable, local
bsc-m03 0.5.5 executable, and Python-backed gzip-9, bzip2-9, xz-9, brotli-11,
and zstd-19 implementations.  No PAQ or CMIX number is reported because no
paq8/paq8px or cmix executable was available in the local environment; this is
recorded as a blocked baseline, not as a measured result.

\subsection{Hardware and implementation}

The local verification used Windows with MinGW g++ 15.2.0, an Intel Core
i7-11700F CPU with 8 cores and 16 logical processors, and about 31.8 GiB
visible physical memory.  Full encode/decode profiling used the local
native executable built from the STC C/C++ implementation.

\subsection{Main enwik9 result}

Table~\ref{tab:stc-score} gives the local LTCB-style accounting for STC.  The
archive and decoder package are fixed outputs from the
verified run.  The table separates archive bytes from decoder/source-package
bytes because LTCB-style accounting charges both: an archive is only useful as
a valid local submission if an independently buildable decompressor accompanies
it.  Transport packaging used for delivery is not part of this score.

\begin{table}[t]
\centering
\small
\begin{tabular}{lrl}
\toprule
Item & Bytes & Scored \\
\midrule
Archive & 157,388,188 & yes \\
Decoder source package & 183,174 & yes \\
Local total & 157,571,362 & yes \\
\bottomrule
\end{tabular}
\caption{Local LTCB-style accounting for STC on enwik9.}
\label{tab:stc-score}
\end{table}

\subsection{Same-coder ablation}

Table~\ref{tab:full-ablation} is the core experimental result.  All rows use
the same internal component coder on full enwik9; the changed factor is the
digit-context decomposition.  The raw digit side stream already improves over
the no-split control.  Bucket packing and bucket ordering both provide larger
gains, and the full STC decomposition is best, reducing the archive by
2,629,561 bytes relative to the no-split control.

The rows isolate the role of each transform choice.  The no-split row is the
backend control: all bytes remain in one component and the component coder is
asked to handle digit variation inside the main stream.  The raw digit side
stream removes the digits but stores their values with little additional
structure; its 469,013-byte improvement shows that separation alone is already
useful.  Turning on bucket packing while leaving sorting off improves the
archive by 1,902,759 bytes, indicating that compact representation of digit
payloads is a major part of the gain.  Turning on bucket sorting while leaving
packing off improves by 2,410,872 bytes, showing that context-conditioned
ordering is also strong.  The full STC row combines both effects and is the
best measured variant.

\begin{table}[t]
\centering
\small
\begin{tabular}{lrr}
\toprule
Variant & Archive bytes & Delta vs no split \\
\midrule
Same coder, no digit split & 160,017,749 & 0 \\
Raw digit side stream & 159,548,736 & -469,013 \\
Bucket packing, sorting off & 158,114,990 & -1,902,759 \\
Bucket sorting, packing off & 157,606,877 & -2,410,872 \\
Full STC & 157,388,188 & -2,629,561 \\
\bottomrule
\end{tabular}
\caption{Full-enwik9 ablation with the same component coder.}
\label{tab:full-ablation}
\end{table}

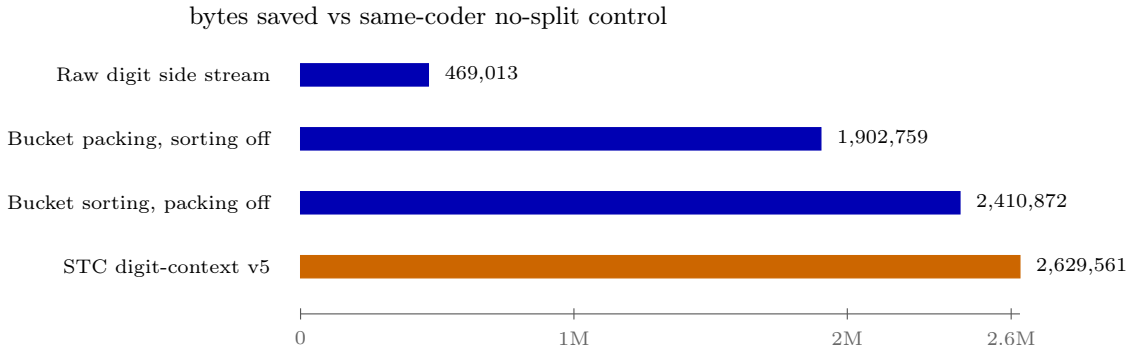
\begin{figure}[t]
\centering
\resizebox{0.92\linewidth}{!}{\begin{tikzpicture}[font=\scriptsize, x=1cm, y=1cm]
\node[font=\small] at (5.1,4.05) {bytes saved vs same-coder no-split control};

\node[anchor=east] at (3.25,3.32) {Raw digit side stream};
\draw[fill=blue!70!black, draw=blue!70!black] (3.5,3.2) rectangle ++(1.60,0.28);
\node[anchor=west] at (5.18,3.34) {469,013};

\node[anchor=east] at (3.25,2.52) {Bucket packing, sorting off};
\draw[fill=blue!70!black, draw=blue!70!black] (3.5,2.4) rectangle ++(6.51,0.28);
\node[anchor=west] at (10.09,2.54) {1,902,759};

\node[anchor=east] at (3.25,1.72) {Bucket sorting, packing off};
\draw[fill=blue!70!black, draw=blue!70!black] (3.5,1.6) rectangle ++(8.25,0.28);
\node[anchor=west] at (11.83,1.74) {2,410,872};

\node[anchor=east] at (3.25,0.92) {STC digit-context v5};
\draw[fill=orange!80!black, draw=orange!80!black] (3.5,0.8) rectangle ++(9.00,0.28);
\node[anchor=west] at (12.58,0.94) {2,629,561};

\draw[gray!70!black] (3.5,0.35) -- (12.5,0.35);
\foreach \x/\lab in {3.5/0,6.92/1M,10.34/2M,12.39/2.6M} {
  \draw[gray!70!black] (\x,0.28) -- (\x,0.42);
  \node[font=\scriptsize, text=gray!80!black] at (\x,0.05) {\lab};
}
\end{tikzpicture}}
\caption{Full-enwik9 same-coder ablation.  Each row changes only the
digit-context decomposition while keeping the component coder fixed.}
\label{fig:ablation}
\end{figure}

\subsection{Component accounting}

The archive component accounting is:
\[
  157{,}388{,}188
  = 43 + 148{,}105{,}611 + 9{,}282{,}534.
\]
Table~\ref{tab:component-accounting} gives component-level raw and compressed
payload sizes.  Raw component sizes are pre-coding stream sizes, not a
partition of the input: the main stream has the original byte length after
digit replacement, while the side streams contain the removed digit payloads
required for reconstruction.

The component accounting also shows why the transform is not merely moving bytes from one
place to another.  The normalized main stream compresses to 148,105,611 bytes,
and all digit side streams together compress to 9,282,534 bytes.  The largest
side streams are the short-run buckets, especially bucket 2 and bucket 3,
which is consistent with structured text containing many small numeric
fragments.  Bucket 4 is split into two components because its two packed bytes
have different useful orderings.  Longer-run buckets are small in total but
remain necessary for exact reconstruction.

\begin{table}[t]
\centering
\scriptsize
\begin{tabular}{lrrp{0.34\linewidth}}
\toprule
Component & Raw bytes & Compressed bytes & Notes \\
\midrule
$m$ & 1,000,000,000 & 148,105,611 & Placeholder-normalized main stream \\
$s_0$ & 3,507,416 & 1,028,681 & Bucket 1 stream \\
$s_1$ & 4,761,161 & 2,727,656 & Bucket 2 stream \\
$s_2$ & 3,126,704 & 1,473,594 & Bucket 3 stream \\
$s_3$ & 2,648,434 & 232,291 & Bucket 4 first-byte stream \\
$s_4$ & 2,648,434 & 1,643,310 & Bucket 4 second-byte stream \\
$s_5$ & 793,695 & 426,826 & Bucket 5 stream \\
$s_6$ & 1,138,644 & 787,420 & Bucket 6 stream \\
$s_7$ & 173,768 & 68,389 & Bucket 7 stream \\
$s_8$ & 1,020,848 & 766,332 & Bucket 8 stream \\
$s_9$ & 39,213 & 14,817 & Bucket 9 stream \\
$s_{10}$ & 125,820 & 90,550 & Bucket 10 stream \\
$s_{11}$ & 5,148 & 1,857 & Bucket 11 stream \\
$s_{12}$ & 61,207 & 20,811 & Bucket 12, length $\ge 12$ \\
\midrule
Header & N/A & 43 & Mode and 14 compressed lengths \\
Total side & 20,050,492 & 9,282,534 & Sum of $s_0,\ldots,s_{12}$ \\
Archive & N/A & 157,388,188 & Equals file-system byte count \\
\bottomrule
\end{tabular}
\caption{Component-level archive accounting.}
\label{tab:component-accounting}
\end{table}

\subsection{Runtime and memory}

STC is a compression-ratio-first prototype, not a speed-oriented compressor.
Table~\ref{tab:runtime} gives the full-enwik9 local runtime profile.  The measurements
wrapper recorded the process identifier, command, stdout/stderr paths, peak
working set, and elapsed wall time.

These measurements characterize the prototype rather than establishing a
speed claim.  The implementation favors simple full-component processing and
large in-memory structures, so it uses roughly 12 GiB of peak memory and about
eleven minutes per direction on the local machine.  A production compressor
would need additional engineering for cache behavior, streaming, parallel
component processing, and bounded-memory suffix construction.

\begin{table}[t]
\centering
\small
\begin{tabular}{lrr}
\toprule
Run & Time & Peak RAM \\
\midrule
Full encode & 693.646 s & 12.491 GiB \\
Full decode & 669.796 s & 12.315 GiB \\
\bottomrule
\end{tabular}
\caption{Runtime and memory for local full-enwik9 profiling.}
\label{tab:runtime}
\end{table}

\subsection{Contextual bsc-m03 comparison}

Table~\ref{tab:bsc-context} provides context but is not the central claim.  The
local bsc-m03 run is not a rule-equivalent LTCB comparison because
decoder-package accounting differs.  The official LTCB-page value \citep{ltcb}
is cited as external context only; the STC result itself is the same-coder ablation and local score accounting
in Table~\ref{tab:stc-score}.

This distinction matters because compression benchmarks are sensitive to what
is counted.  The same-coder ablation in Table~\ref{tab:full-ablation} is the
paper's controlled evidence for the transform.  The bsc-m03 row only places
the local total in a familiar BWT-family neighborhood and should not be read
as an official leaderboard comparison.

\begin{table}[t]
\centering
\small
\begin{tabular}{lrl}
\toprule
Comparator & Bytes & Status \\
\midrule
STC local total & 157,571,362 & archive plus decoder package \\
Local bsc-m03 0.5.5 comparison & 160,364,392 & contextual, not rule-equivalent \\
Official LTCB-page bsc-m03 score & N/A & leaderboard context only \\
\bottomrule
\end{tabular}
\caption{Contextual bsc-m03 comparison.}
\label{tab:bsc-context}
\end{table}

\subsection{Supplementary corpus runs}

Supplementary corpus runs expose both portability signals and current
robustness failures; because two STC rows are partial, they are not used as
primary compression comparisons.  The full local matrix and blocked-row
accounting are treated as supplementary release material rather than as part of the main claim.

The supplementary results are therefore diagnostic.  They help identify where
the current encoder is mature enough for roundtrip experiments and where
parser or resource-handling work remains.  The main conclusion does not depend
on those rows.

\section{Discussion}

\subsection{Reproducibility and claim boundary}

The enwik9 result is a local, source-contained compression result rather than
an official leaderboard entry.  The relevant reproducibility requirements are
stable and paper-facing: the public input must match the reference enwik9 hash,
the compressed archive must decode to the exact 1,000,000,000-byte input, the
counted decoder source package must be buildable by an external reader, and the
reported local total must equal the archive bytes plus the decoder-source bytes.
Official benchmark status would require independent maintainer-side rebuild,
decode, hash verification, and score acceptance.

The implementation records more machine-readable evidence than is appropriate
for the manuscript itself.  Detailed release files, including exact file
hashes, decoder-source packaging, build commands, and verification logs, should
be published with the code release or dataset page.  The paper uses those data
to support the stated result, but it does not present internal run paths,
internal process files, or engineering version labels as part of the scientific
method.

\subsection{Scope of the evidence}

The central causal comparison is the same-coder ablation on full enwik9.  This
comparison holds the component coder fixed and changes only the digit-context
decomposition, making it the cleanest evidence for the contribution of STC.
The bsc-m03 and standard-codec measurements provide context, but they are not
leaderboard claims because decoder-package accounting, executable provenance,
and independent reproduction status differ across systems.

The current evidence is also corpus-specific.  Supplementary runs on enwik8,
Calgary, Canterbury, and Silesia are useful diagnostics, but two STC rows are
partial because the prototype encoder did not complete every file.  The method
therefore should not be described as a universal BWT theorem or a general
compression improvement across all data.  Broader text, markup, source-code,
log, binary, and adversarial datasets are needed before making that claim.

\subsection{Engineering limitations}

The current implementation is a compression-ratio-first research prototype.
Full enwik9 encode/decode uses roughly 12 GiB peak RAM and takes about eleven
minutes per direction on the local machine.  These numbers are acceptable for a
research prototype but not a final engineering claim.  Parser robustness,
malformed-archive fuzzing, memory reduction, and independent code review remain
necessary before STC should be treated as production compression software.

\section{Conclusions and Implications}

STC is presented as a new algorithm found by the authors through the
self-evolving AI system zeelin.  It shows that a reversible byte-level
preprocessing step can improve a BWT-family compressor on structured text when
the backend coder is held fixed.  The method separates digit payloads from the
textual neighborhoods that locate them, orders the side-stream records using
decoder-visible context, and restores the original file without a serialized
permutation.  On full enwik9, the
same-coder ablation attributes a 2,629,561-byte archive reduction to the full
STC transform relative to the no-split control.

The result should be interpreted with the same restraint used throughout the
paper.  It is a local compression result with source-contained decoding and
hash-based reconstruction checks, not an official leaderboard entry and not a
new theorem about BWT compressiveness.  The strongest evidence is the controlled
ablation on enwik9; supplementary corpus runs and external compressors provide
context but do not replace independent benchmark verification.

The next step is engineering and independent reproduction rather than changing
the claim.  A public release should provide the source code, decoder source
package, scored archive, verification hashes, and exact build/decode commands in
the repository or linked dataset page.  Broader corpus coverage, malformed-input
fuzzing, memory reduction, and independent code review are needed before STC can
be treated as a general-purpose compression system.

\section*{Data Availability Statement}

The experiments use enwik9, the public 1,000,000,000-byte
Wikipedia-derived corpus used by the Large Text Compression Benchmark.  The
paper reports SHA-256 hashes for the input, the compressed archive, and the
decoder source package so that readers can verify that they are using the same
release files.  The STC source code, decoder source package, build
instructions, paper source, figure sources, verification summaries, release
notes, and supplementary tables are provided through the
\href{https://github.com/thu-nmrc/STC-for-BWT-FamilyText-Compression}{THU-NMRC
project repository}.  Large generated release files, including the scored
archive, verification logs, hashes, and supplementary corpus outputs, will be
distributed as GitHub Release assets or through a linked Hugging Face Dataset
record rather than committed directly to Git.  The release README will map each
paper table to its source file or script, record the exact build and decode
commands, and state which files are required for independent reproduction.

\section*{Conflict of Interest Statement}

The authors report no conflict of interest for this arXiv version.

\bibliographystyle{plainnat}
\bibliography{references}

\end{document}